\documentclass[12pt,onecolumn]{article}
\usepackage[fleqn]{amsmath}
\usepackage{mathtools}
\usepackage{amsthm}

\usepackage{algorithmic}

\usepackage{array}

\usepackage[caption=false,font=footnotesize,labelfont=sf,textfont=sf]{subfig}

\usepackage{fixltx2e}
\usepackage{stfloats}
\usepackage{url}
\usepackage{mathrsfs}
\usepackage{bbold}
\usepackage[linesnumbered, boxed, ruled,vlined]{algorithm2e}
\usepackage{soul}
\usepackage{xcolor}

\usepackage[shortlabels]{enumitem}

\newtheorem{theorem}{Theorem}[section]

\newtheorem{proposition}[theorem]{Proposition}

\providecommand{\keywords}[1]{\textbf{\textit{Index terms---}} #1}

\begin{document}

\title{Sequence Covering Similarity for Symbolic Sequence Comparison}

\author{Pierre-Francois Marteau\\
IRISA, Universite Bretagne Sud}




\maketitle


\begin{abstract}
This paper introduces the sequence covering similarity, that we formally define for evaluating the similarity between a symbolic sequence (string) and a set of symbolic sequences (strings). From this covering similarity we derive a pair-wise distance to compare two symbolic sequences.  
We show that this covering distance is a semimetric. Few examples are given to show how this string semimetric in $O(n \cdot log n)$ compares with the Levenshtein's distance that is in $O(n^2)$. A final example presents its application to plagiarism detection. 
\end{abstract}


%

\section{Introduction}\label{sec:introduction}
Estimating efficiently the similarity between symbolic sequences is a recurrent task in various application domains, in particular in bio-informatics, text processing or computer or network security. Numerous similarity measures have been defined to cope with symbolic sequences such the edit distance and its implementation proposed by Wagner and Fisher \cite{Wagner1974}, BLAST \cite{Korf2003}, the Smith and Waterman or Levenshtein \cite{SMITH1981, Levenshtein66} and the  Needleman Wunch \cite{NEEDLEMAN1970} distances or the local sequence kernels \cite{Vert2004}.   

We present in this paper a new approach to characterize similarity between sequences by introducing the notion of sequence covering. Basically, this similarity is based on a set of reference sequences which defines a dictionary of subsequences that are used to 'optimally' cover any sequence. Originally this sequence covering principle has been introduced in the context of Host Intrusion Detection \cite{Marteau2017}. We derive hereinafter a pairwise similarity measure and show that this measure is a semimetric on the set of strings. We finally highlights through some examples the utility of this measure.

\section{The Sequence Covering Similarity}

\begin{figure}[ht!]
\centering
\includegraphics[scale=.7]{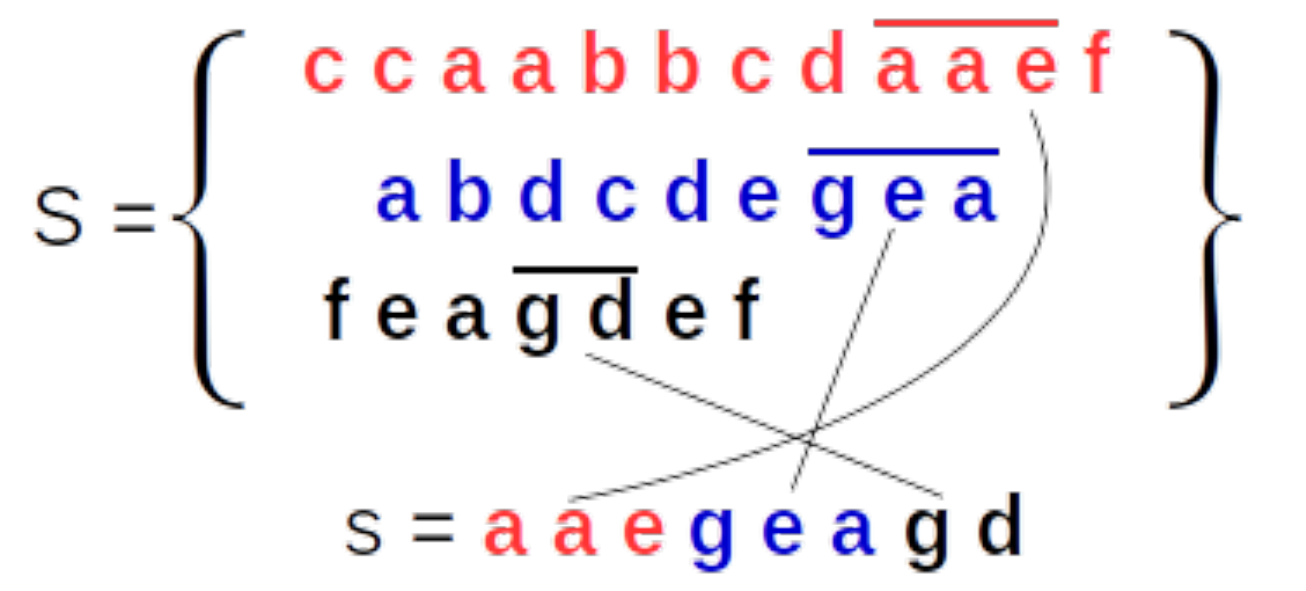}
\caption{Example of the covering of a sequence ($s$) using subsequences of sequences in a set ($S$).}
\label{fig:principle}
\end{figure}

The notion of sequence covering is simple and depicted in Fig. \ref{fig:principle}. The sequence $s$ is \textit{covered} by subsequences of the sequences that belong to set $S$. On this example, the covering is \textit{optimal} in the sense that it is composed with a minimal number of subsequences. It is \textit{total} in the sense that all the elements of $s$ are \textit{covered}. 

The sequence covering similarity between $s$ and set $S$ relates the size (in number of subsequences) of the \textit{optimal} covering of $s$ using sequences of $S$, to the size of $s$ (in number of elements) itself, $|s|$, such that it is maximum equal to one if the covering is of size 1, and minimal equal to $1/|s|$ if the covering is composed with subsequences of size 1.

We define precisely these notions in the following subsection.

\subsection{Definitions and notation}

Let $\Sigma$ be a finite alphabet and let $\Sigma^*$ be the set of all sequences (or strings) define over $\Sigma$. We note $\epsilon$ the empty sequence.

Let $S \subset \Sigma^*$ be any set of sequences, and let $S_{sub}$ be the set of all subsequences that can be extracted from any element of $S\cup \Sigma$. We denote by $\mathscr{M}(S_{sub})$ the set of all the multisets\footnote{A multiset is a collection of elements in which elements are allowed to repeat; it may contain a finite number of indistinguishable copies of a particular element.} that we can compose from the elements of $S_{sub}$.

$c \in \mathscr{M}(S_{sub})$ is called a partial covering of sequence $s \in \Sigma^*$ iif 
\begin{enumerate}
\item all the subsequences of $c$ are also subsequences of $s$,
\item indistinguishable copies of a particular element in $c$ correspond to distinct occurrences of the same subsequence in $s$.
\end{enumerate}

If $c \in \mathscr{M}(S_{sub})$ entirely covers $s$, meaning that we can find an arrangement of the elements of $c$ that covers entirely $s$, then we will call it a full covering for $s$.


Finally, we call a $S$-optimal covering of $s$ any full covering of $s$ which is composed with a minimal number of subsequences in $S_{sub}$.

Let $c^*_S(s)$ be a $S$-optimal covering of  $s$.

We define the covering similarity measure between any non empty sequence $s$ and any set $S \subset \Sigma^*$ as
\begin{equation}
\mathscr{S}(s,S) = \frac{|s|-|c^*_S(s)|+1}{|s|} 
\label{eq:coveringSimilarity}
\end{equation}
where $|c^*_S(s)|$ is the number of subsequences composing a $S$-optimal covering of $s$, and $|s|$ is the length of sequence $s$.\\

Note that in general $c^*_S(s)$  is not unique, but since all such coverings have the same cardinality, $|c^*_S(s)|$, $\mathscr{S}(s,S)$ is well defined. 

Properties of $\mathscr{S}(s,S)$:
\begin{enumerate}
\item If $s$ is a non empty subsequence in $S_{sub}$, then  $\mathscr{S}(s,S) = 1$ is maximal.
\item In the worse case, the $S$-optimal covering of $s$ has a cardinality equal to $|s|$, meaning that it is composed only with subsequences of length $1$. In that case, $\mathscr{S}(s,S) = \frac{1}{|s|}$ is minimal.
\item If $s$ is non empty, $\mathscr{S}(s,\emptyset) = \frac{1}{|s|}$ (notice that if $S=\emptyset$, $S_{sub}=\Sigma$).
\end{enumerate} 

Furthermore, as $\epsilon$ is a subsequence of any sequence in $\Sigma^*$, we define, for any set $S \subset \Sigma^*$, $\mathscr{S}(\epsilon,S) = 1.0$

As an example, let us consider the following case:

\begin{flalign*} 
&s_1 = [0,0,0,0,1,1,1,1,0,0,0,0,1,1,1,1] \\
&s_2 = [0,0,0,0,0,0,0,0,1,1,1,1,1,1,1,1]\\
&S= \{s_1, s_2\}\\
&s_3 = [0,0,1,1,0,0,1,1,0,0,1,1,0,0,1,1]\\
&s_4 = [0,1,0,1,0,1,0,1,0,1,0,1,0,1,0,1]\\
\end{flalign*} 

The $S$-optimal covering of $s_3$ \footnote{([0,0,1,1][0,0,1,1],[0,0,1,1][0,0,1,1]) is a $S$-optimal covering of $s_3$} is of size $4$, hence $\mathscr{S}(s_3,S) = \frac{16-4+1}{16}=13/16$, and the $S$-optimal covering of $s_4$ \footnote{([0,1],[0,1],[0,1],[0,1],[0,1],[0,1],[0,1],[0,1]) is a  $S$-optimal covering for $s_4$} is of size $8$, leading to $\mathscr{S}(s_4,S) = \frac{16-8+1}{16}=9/16$.\\


\subsection{Finding a $S$-optimal covering for any tuple $(s,S)$}

The brute-force approach to find a $S$-optimal covering for a sequence $s$ is presented in algorithm \ref{alg:Find-S-optimal}. It is an incremental algorithm that, first, finds the longest subsequence of $s$ that is contained in $S_{sub}$ and that starts at the beginning of $s$. This first subsequence is the first element of the $S$-optimal covering. Then, it searches for the following longest subsequence that is in $S_{sub}$ and that starts at the end of the first element of the covering, adds it to the covering in construction, and iterate until reaching the end of sequence $s$.  

\begin{algorithm}
\SetKwData{Left}{left}\SetKwData{This}{this}\SetKwData{Up}{up}
\SetKwFunction{Union}{Union}\SetKwFunction{FindCompress}{FindCompress}
\SetKwInOut{Input}{input}\SetKwInOut{Output}{output}
\Input{$S \subset \Sigma^*$, a set of sequences}
\Input{$s \in \Sigma^*$, a test sequence }
\Output{$c$, a ($S$-optimal) covering for $s$}
\BlankLine
$continue \longleftarrow True$\;
$start \longleftarrow 0$\;
$c^* \longleftarrow \emptyset$\;
\While{continue}{
  $end \longleftarrow start + 1$\;
  \While{$end<|s|$ and $s[start:end] \in S_{sub}$}{
  $end \longleftarrow end + 1$\;
  }
  $c \longleftarrow c^* \cup \{s[start:end-1]\}$\;
  \lIf{$end = |s|$}{$continue \longleftarrow False$}
  $start \longleftarrow end $\;
}
\Return $c$\;
\BlankLine
\caption{Find a $S$-optimal covering for $s$\label{alg:Find-S-optimal}}
\end{algorithm}

\begin{proposition} Algorithm \ref{alg:Find-S-optimal} outputs a $S$-optimal covering for sequence $s$.
\end{proposition}

\begin{proof} i) First we notice that since all the subsequences of length $1$ constructed on $\Sigma$ are included into $S_{sub}$, algorithm \ref{alg:Find-S-optimal}, by construction, necessarily outputs a full covering of $s$ (meaning that $s$ is entirely covered by the subsequences of the covering provided the algorithm).  

ii) Second we notice that, for all $s_1$ and $s_2$ in $\Sigma^*$ such that $s_1$ is a subsequence of $s_2$, and any $S \subset \Sigma^*$, $|c^*_S(s_1)| \le |c^*_S(s_2)|$.

We finalize the proof by induction on $n$, the cardinality (the size) of the coverings. 

The proposition is obviously true for $n=1$: for all sequence $s$ for which a covering of size $1$ exists (meaning that $s$ is a subsequence of one of the sequences  in $S$), algorithm \ref{alg:Find-S-optimal} finds the S-optimal covering that consists of $s$ itself.

Then, assuming that the proposition holds for  $n$, such that $n \ge 1$ (IH),  we consider a sequence $s$ that admits a $S$-optimal covering of size $n+1$.

Let $s=s_1+\overline{s}_1$, be the decomposition of $s$ according to the full covering provided by algorithm \ref{alg:Find-S-optimal}, where $s_1$ is the prefix of the covering (first element) and $\overline{s}_1$ the remaining suffix subsequence (concatenation of the remaining covering elements). $+$ is the sequence concatenation operator. Similarly, let $s=s^*_1+\overline{s}^*_1$, be the decomposition of $s$ according to a $S$-optimal covering of $s$. Necessarily, $s^*_1$, which is also a prefix of $s$, is a subsequence of $s_1$ (otherwise, since $s^*_1$ is in $S_{sub}$, algorithm \ref{alg:Find-S-optimal} would have increased the length of $s_1$ at least to the length of $s^*_1$). Hence, $\overline{s}_1$ is a subsequence of $\overline{s}^*_1$ and, according to ii), $|c^*_S(\overline{s}_1)| \le |c^*_S(\overline{s}^*_1)|=n$. This shows that $\overline{s}_1$ is a sequence that admits a $S$-optimal covering, $c^*_S(\overline{s}_1)$,  of size at most equal to $n$. According to (HI), algorithm \ref{alg:Find-S-optimal} returns such an optimal covering for  $\overline{s}_1$. This shows that the covering $\{s_1\} \cup c^*_S(\overline{s}_1)$ that is returned by algorithm \ref{alg:Find-S-optimal} for the full sequence $s$, is at most of size $n+1$, meaning that it is actually a  $S$-optimal covering for $s$ of size $n+1$. Hence, by induction, the proposition is true for all $n$, which proves the proposition.
\end{proof}

\subsubsection{Other property}

\begin{proposition}
By definition of the $S$-optimal covering of a sequence, it is easy to show that \\ 

For all $S \subset \Sigma^*$, all $A \subset S$ and all $s \in \Sigma^*$, $|c^*_S(s)| \leq |c^*_A(s)|$, leading to $\mathscr{S}(s,S) \ge \mathscr{S}(s,A)$.
\end{proposition}

\subsection{Pairwise similarity and pairwise distance for comparing pairs of symbolic sequences (strings)}
The covering similarity between a sequence and a set of sequences as defined in Eq. \ref{eq:coveringSimilarity} allows for the definition of a covering similarity measure on the sequence set, $\Sigma^*$, itself. For any pair of non empty sequences $s_1$, $s_2 \in \Sigma^*$ we define it as follows 
\begin{equation}
\mathscr{S}_{seq}(s_1,s_2) = \frac{1}{2} (\mathscr{S}(s_1,\{s_2\})+ \mathscr{S}(s_2,\{s_1\})  )
\label{eq:SeqCoveringSimilarity}
\end{equation}

where $\mathscr{S}$ is defined in Eq. \ref{eq:coveringSimilarity}.

Then, we define $\mathscr{S}_{seq}(\epsilon,\epsilon) = 1.0$, and for any non empty $s \in 
\Sigma^*$, we get that 
$\mathscr{S}_{seq}(\epsilon, s) = \mathscr{S}_{seq}(s, \epsilon) = \frac{1}{2} (1+\frac{1}{|s|+1})$\\

Finally we define straightforwardly $\delta_{c}$ a pairwise distance on $\Sigma^*$ as
\begin{align}
\delta_{c}(s_1,s_2) = 1 - \mathscr{S}_{seq}(s_1,s_2) 
\label{eq:SeqCoveringDistance}
\end{align}

Leading to 
\begin{align}
\label{eq:epsilonDistance1}
\delta_{c}(\epsilon,\epsilon) = 0 \texttt{ and,} \\ 
\texttt{for any non empty } s \in \Sigma^* \texttt{,  } \delta_{c}(\epsilon,s) = \delta_{c}(s, \epsilon) = \frac{1}{2} (1 - \frac{1}{|s|+1})  \nonumber
\end{align}

\begin{proposition} $\delta_c(.,.)$ is a semimetric on $\Sigma^*$
\end{proposition}

\begin{proof} 
\begin{flushleft}
It is easy to verify that $\delta_c$ is \textbf{non negative}: for all $s \in \Sigma^*$, and all $S \subset \Sigma^*$, $\mathscr{S}(s,S) \in [\frac{1}{|s|+1}; 1]$. Hence, for all $s_1$, $s_2 \in \Sigma^*$, $\delta_{c}(s_1,s_2) \in [\frac{1}{|2|}\cdot(\frac{1}{|s_1|+1}+\frac{1}{|s_2|+1}); 1]$, and, according to Eq. and Eq. \ref{eq:SeqCoveringDistance} \ref{eq:epsilonDistance1},  for all  $s_1$, $s_2 \in \Sigma^*$, $\mathscr{S}_{seq}(s_1,s_2) \in [0; 1]$.\\
\end{flushleft}

\begin{flushleft}
\textbf{identity of indiscernibles}: First, for all  $s_1$, $s_2 \in \Sigma^*$, if $s_1 = s_2$, then $\mathscr{S}(s_1,\{s_1\})=1$ hence $\delta_{c}(s_1,s_2)=0$.\\
Conversely, for all  $s_1$, $s_2 \in \Sigma^*$ s.t. $\delta_{c}(s_1,s_2)=0$, 
\begin{itemize}
\item if $s_1=\epsilon$, then necessarily $s_2=\epsilon$, otherwise $|s_2|>0$ and  $\delta_{c}(\epsilon,s_2)=\frac{1}{2}(1-\frac{1}{|s_2|+1}) > 0$
\item If if $s_1 \neq\epsilon$, then necessarily  $s_2 \neq\epsilon$ and, since \\
$\delta_{c}(s_1,s_2)=1 - \frac{1}{2}(\mathscr{S}(s_1,\{s_2\})+ \mathscr{S}(s_2,\{s_1\}))=0$, necessarily  $\mathscr{S}(s_1,\{s_2\})=\mathscr{S}(s_2,\{s_1\})=1$, which means that $s_1$ is a subsequence of $s_2$ and conversely, $s_2$ is a subsequence of $s_1$, showing that $s_1=s_2$.
\end{itemize}
\end{flushleft}

\begin{flushleft}
\textbf{symmetry}: As $\mathscr{S}_{seq}(.,.)$ is symmetric by construction, so is $\delta_{c}(.,.)$.
\end{flushleft}

\end{proof} 

\section{Algorithmic complexity}
A suffix tree implementation of algorithm \ref{alg:Find-S-optimal} leads to a time complexity that is upper bounded by $O(k\cdot |s| \cdot log(|s|))$, where $k=c^*_S({s})$ is the size of a $S$-optimal covering for $s$. \\

The previous time complexity does not depend on $|S|$, which means that we can increase the size of $S$ without loosing on the processing time. This property is particularly important for applications for which $|S|$ is potentially large such as in plagiarism detection for instance.\\

For the pairwise distance $\delta_{c}(s_1,s_2)$, the time complexity is $O(k_1\cdot |s_1| \cdot log(|s_1|) + k_2\cdot |s_2| \cdot log(|s_2|))$ where $k_1=c^*_{\{s_2\}}({s_1})$ is the size of a $\{s_2\}$-optimal covering for $s_1$ and $k_2=c^*_{\{s_1\}}({s_2})$ is the size of a $\{s_1\}$-optimal covering for $s_2$. In comparison, the Levenshtein's distance is in $O(|s|^2)$.

\section{Examples}
We give below some examples that present the use of the covering similarity or distance for string matching and processing. A python 3 implementation available at  \url{https://github.com/pfmarteau/STree4CS} allows to play these examples. 

\subsection{Pairwise distances on strings}
Table \ref{tab:ex} presents the covering distance values obtained for some pairs of strings. As a comparative baseline, the Levenshtein's distance \cite{Levenshtein66} is also given for the same pairs of strings.

\begin{table}[!h]
\center
\resizebox{\linewidth}{!}{
\begin{tabular}{cc||cc}
\hline\hline
$string_1$ & $string_2$ & $\delta_{c}$ & Levenshtein \footnote{Levenshtein, Vladimir I. (February 1966). "Binary codes capable of correcting deletions, insertions, and reversals". Soviet Physics Doklady. 10 (8): 707–710}\\
\hline
\textbf{'amrican'} & \textbf{'american'} & .196 & \textbf{.067}\\
\textbf{'european'}& \textbf{'american'}& \textbf{.75} & .375\\
'european'& 'indoeuropean'& .167 & .25\\
\textbf{'indian'}& \textbf{'indoeuropean'}&  .5 & \textbf{.583}\\
'indian'& 'american'& 0.708 & .417\\
'narcotics' & 'narcoleptics' & .222 & .167\\
'\textbf{'little big man'}& \textbf{'big little man'}& \textbf{.143} & .286\\
\hline
\end{tabular}}
\caption{Covering and Levenshtein's distances on some pairs of strings. Min and max values for each distance are in bold fonts.}
\label{tab:ex}
\end{table}

\subsection{Detection of plagiarism}

We show in this example how the sequence covering similarity is able to detect lifted passage of an original source text spread in a plagiarized text.\\
 
This example (Example 2) is borrowed from \\
\url{https://www.princeton.edu/pr/pub/integrity/pages/plagiarism/}\\

\textbf{Original source text}\\
\textit{"From time to time this submerged or latent theater in Hamlet becomes almost overt. It is close to the surface in Hamlet’s pretense of madness, the “antic disposition” he puts on to protect himself and prevent his antagonists from plucking out the heart of his mystery. It is even closer to the surface when Hamlet enters his mother’s room and holds up, side by side, the pictures of the two kings, Old Hamlet and Claudius, and proceeds to describe for her the true nature of the choice she has made, presenting truth by means of a show. Similarly, when he leaps into the open grave at Ophelia’s funeral, ranting in high heroic terms, he is acting out for Laertes, and perhaps for himself as well, the folly of excessive, melodramatic expressions of grief."}\\

\textbf{Plagiarism: Lifting selected passages and phrases without proper acknowledgment (lifted passages are underlined)}\\ 
\textit{"Almost all of Shakespeare’s Hamlet can be understood as a play about acting and the theater. For example, in Act 1, Hamlet adopts a \ul{pretense of madness} that he uses \ul{to protect himself and prevent his antagonists from} discovering his mission to revenge his father’s murder. He also presents \ul{truth by means of a show} when he compares the portraits of Gertrude’s two husbands in order \ul{to describe for her the true nature of the choice she has made}. And when he leaps in Ophelia’s open grave \ul{ranting in high heroic terms}, Hamlet is \ul{acting out the folly of excessive, melodramatic expressions of grief}"}.\\

\noindent\texttt{Covering Distance = 0.219}\\
\texttt{Covering Similarity = 0.801}\\
\texttt{Covering = ['A', 'lmost ', 'al', 'l', ' of ', 'S', 'ha', 'k', 'es', 'pe', 'ar', 'e', '’s ', 'Hamlet ', 'c', 'an', ' be', ' u', 'nd', 'ers', 'to', 'od', ' as ', 'a ', 'pl', 'a', 'y ', 'a', 'b', 'out ', 'acting ', 'and ', 'the t', 'heater', '. ', 'F', 'or ', 'ex', 'am', 'pl', 'e, ', 'in ', 'A', 'ct ', '1', ', ', 'Hamlet a', 'd', 'op', 'ts ', 'a ', \ul{'pretense of madness'}, ' th', 'at ', 'he ', 'us', 'es ', \ul{'to protect himself and prevent his antagonists from '}, 'dis', 'co', 'ver', 'ing ', 'his m', 'is', 'sion', ' to ', 'reven', 'ge', ' his ', 'fa', 'ther’s ', 'm', 'ur', 'de', 'r', '. ', 'H', 'e a', 'l', 's', 'o pr', 'esent', 's \textcolor{red}{\ul{t}}', \ul{'ruth by means of a show'}, ' when he ', 'com', 'p', 'ar', 'es ', 'the p', 'or', 'tr', 'a', 'it', 's of ', 'G', 'ert', 'ru', 'de', '’s ', 'two ', 'h', 'us', 'b', 'and', 's in', ' or', 'de', 'r \textcolor{red}{\ul{to }}', \ul{'describe for her the true nature of the choice she has made'}, '. ', 'A', 'nd ', 'when he leaps in', ' Ophelia’s ', 'open grave ', \ul{'ranting in high heroic terms, '}, 'Hamlet ', \ul{'is acting out '}, \ul{'the folly of excessive, melodramatic expressions of grief.'}]}\\

The small differences between lifted passages that are underlined in the original text and in the optimal covering is due to the non-uniqueness of the optimal covering. A simple post-processing can easily correct these small discrepancies. We notice also that few covering substrings such as 'when he leaps in' have not been underlined in the original text.\\

Indeed, if the plagiarized text is re-written with the same text structure but different wording, then the similarity would drop, and the covering won't be so informative.

\section{Conclusion}
We have introduced the notion of sequence covering given a set of reference sequences which define a dictionary of subsequences that are used to 'optimally' cover any sequence. Originally this notion has been introduced in the context of host intrusion detection. From this notion we have defined a pairwise distance measure that can be used to compare two sequences and shown that this measure is a semimetric. As the nature of the sequence covering similarity is somehow complementary to other existing similarity defined for sequential data, one may conjecture it could help by bringing some complementary discriminant information. In particular, as efficient implementations exist using suffix trees or arrays, this similarity could bring some benefits in bioinformatics or in text processing applications.

%
%
%


\bibliographystyle{IEEEtran}
\bibliography{IEEEabrv,biblio}
%


\end{document}